\documentclass[11pt]{article}
\usepackage{fullpage}

\usepackage{graphics}
\usepackage[dvips]{epsfig}

\usepackage{amsmath}
\usepackage{amssymb}
\usepackage{amsfonts}
\usepackage{graphicx}

\newtheorem{theorem}{Theorem}[section]
\newtheorem{lemma}{Lemma}[section]

\newtheorem{claim}{Claim}[section]

\newtheorem{definition}{Definition}[section]
\newtheorem{fact}{Fact}[section]

\newcommand{\qed}{\hfill $\Box$ \bigbreak}
\newenvironment{proof}{\noindent {\bf Proof.}}{\qed}

\newcommand{\lcm}{\mbox{lcm}}
\newcommand{\dist}{\mbox{dist}}
\newcommand{\caL}{{\cal{L}}}

\newcommand{\remove}[1]{}

\newcommand{\caA}{{\cal{A}}}
\newcommand{\BW}{\mbox{\sc bw}}
\newcommand{\CBW}{\mbox{\sc cbw}}

\newcommand{\prem}{\mbox{\tt prime}}
\renewcommand{\wedge}{\,|\,}

\begin{document}

\baselineskip  0.2in 
\parskip     0.1in 
\parindent   0.0in 

\title{{\bf Delays Induce an Exponential Memory Gap\\ for Rendezvous in Trees}\thanks
{The results of this paper appeared in a preliminary form in papers:
P. Fraigniaud, A. Pelc, Deterministic rendezvous in trees with little memory, Proc. 22nd International Symposium on Distributed Computing (DISC 2008), LNCS 5218, 242-256, and P. Fraigniaud, A. Pelc, Delays induce an exponential memory gap for rendezvous in trees, Proc. 22nd Ann. ACM Symposium on Parallel Algorithms and Architectures (SPAA 2010), 224-232.  }}

\author{
Pierre Fraigniaud\thanks{CNRS, Universit\'e Paris Diderot - Paris 7, France.
E-mail: Pierre.Fraigniaud@liafa.jussieu.fr.
Part of this work was done during
this author's visit at the Research Chair in Distributed Computing of
the Universit\'{e} du Qu\'{e}bec en Outaouais. Additional supports from ANR projects ALADDIN and PROSE, and INRIA project GANG.}
\and
Andrzej Pelc\thanks{D\'{e}partement d'informatique, Universit\'{e} du Qu\'{e}bec en Outaouais,
Gatineau, Qu\'{e}bec J8X 3X7,
Canada. E-mail: pelc@uqo.ca.
Supported in part by NSERC discovery grant 
and by the Research Chair in Distributed Computing of
the Universit\'{e} du Qu\'{e}bec en Outaouais.}
}

\date{ }
\maketitle

\begin{abstract}

The aim of rendezvous in a graph is meeting of two mobile agents at some node of 
an unknown anonymous connected graph.  
In this paper, we focus on rendezvous in trees, and,  analogously to the efforts that have been made for solving the exploration problem with compact automata, we study
the size of memory of mobile agents that permits to solve the
rendezvous problem deterministically.
We assume that the agents are identical,
and move in synchronous rounds. 

We first show that if the delay between the
starting times of the agents is \emph{arbitrary}, then the lower bound on memory required for rendezvous is $\Omega(\log n)$ bits,
even for the line of length $n$. This lower bound meets a previously known upper bound of $O(\log n)$ bits for rendezvous in arbitrary  graphs of
size at most $n$. Our main result is a proof that
the amount of memory needed for rendezvous \emph{with simultaneous start} depends essentially on the number $\ell$ of leaves of the tree, and is exponentially less impacted by the number
$n$ of nodes. Indeed, we present two identical agents with $O(\log \ell + \log\log n)$
bits of memory that solve the rendezvous problem in all trees with at most $n$ nodes and at most $\ell$ leaves.
Hence, for the class of trees with polylogarithmically many leaves, there is an exponential gap in minimum
memory size needed for rendezvous between the scenario
with arbitrary delay and the scenario with delay zero.
Moreover, we show that our upper bound is optimal by proving that $\Omega(\log \ell + \log\log n)$ bits of memory are required for rendezvous, even in the class of trees with degrees bounded
by~3.

\vspace{2ex}

\end{abstract}

\vfill

\thispagestyle{empty}
\setcounter{page}{0}
\pagebreak

\section{Introduction}

The rendezvous in a network~\cite{alpern95a, alpern99} is the following task. Two identical mobile agents, initially located in two nodes of the network,
move along links from node to node, and eventually have to get to the same node at the same time. The network
is modeled as an undirected connected graph,  and agents traverse links in synchronous rounds.
They cannot leave any marks on visited nodes. In this paper we consider
deterministic rendezvous in trees,
and seek rendezvous protocols that do not
rely on the knowledge of node labels, and can work in anonymous trees as well  (cf. \cite{alpern02b}). 
This assumption is motivated by the fact that, even when nodes are equipped with distinct labels, agents may be unable to perceive them,
or nodes may refuse to reveal their labels, e.g., due to security reasons. (Note also that if nodes of the network are labeled using distinct names, then agents can meet at some a priori agreed node, and rendezvous reduces to graph exploration). On the other hand,
edges incident to a node $v$ have distinct labels in 
$\{0,\dots,d-1\}$, where $d$ is the degree of $v$. Thus every undirected
edge $\{u,v\}$ has two labels, which are called its {\em port numbers} at $u$
and at $v$.  (In the absence of port numbers, rendezvous is usually impossible, as the adversary may prevent an agent from taking some edge incident to the current node). A function assigning port numbers to every edge is called a {\em port labeling}. 
Port labeling is {\em local}, i.e., there is no relation between
port numbers at $u$ and at $v$  (we do not assume any sense of direction, of any kind). 

The aim of the present paper is to determine the space complexity of rendezvous in trees. 
We assume that the port labeling is decided by an adversary aiming at preventing two agents from meeting, or at allowing the agents to meet only after having consumed a lot of resources, e.g., memory space. Hence, we adopt the following definition.

\begin{definition}
A pair of agents initially placed at  nodes $u$ and $v$ of 
a tree $T$ solves the rendezvous problem if, for any port labeling of 
$T$, both agents are eventually in the same node of the tree in the 
same round. 
\end{definition}

It is easy to characterize the initial positions $u$ and $v$ of a tree $T$ for which rendezvous is feasible. Recall that an automorphism of the tree is a bijection
$f: V \to V$, where $V$ is the set of nodes of the tree, such that for any $w,w'\in V$, $w$ is adjacent to $w'$ if and only if  $f(w)$ is adjacent to $f(w')$. 
It preserves a given port labeling $\mu$, if for any $w,w'\in V$, the port number corresponding to edge $\{w,w'\}$ at node $w$ is equal to the port
number corresponding to edge $\{f(w),f(w')\}$ at node $f(w)$. 
Nodes $u$ and $v$ of a tree are called {\em topologically symmetric}, if there exists
an automorphism $f$ of the tree, such that $f(u)=v$.
Nodes $u$ and $v$ of a tree with labeling $\mu$ are called {\em symmetric} with respect to this labeling, if there exists
an automorphism $f$ of the tree preserving this port labeling, such that $f(u)=v$.
It is well known (cf., e.g., \cite{CKP}) that rendezvous with simultaneous start in a tree $T$ with a given port labeling $\mu$  is feasible, if and only if the initial positions $u$ and $v$ of agents are not symmetric with respect to this labeling.
Thus the following notion is crucial for our considerations.

\begin{definition}
Nodes $u$ and $v$ of a tree $T=(V,E)$ are {\em perfectly symmetrizable} if there exists a port labeling $\mu$ of $T$
and an automorphism of the tree preserving $\mu$
that carries one node on the other. 
\end{definition}

Note that two nodes  that are perfectly symmetrizable  are necessarily topologically symmetric. On the other hand, two topologically symmetric nodes may not be perfectly symmetrizable. Typical examples are provided by the paths (or lines) with odd numbers of nodes, and by complete binary trees. In both cases, two leaves  are topologically symmetric while they are not perfectly symmetrizable. 

According to the above definitions, one can reformulate the feasibility of rendezvous as follows. 

\begin{fact}
A pair of agents can solve the rendezvous problem in a tree, if and only if their initial positions are not perfectly symmetrizable. 
\end{fact}

Consequently, throughout the paper, we consider only non perfectly symmetrizable initial positions of the agents. 

\subsection{Our results}
\label{subsec:ourresults}

We first show that if the delay between the
starting times of the agents is arbitrary, then the lower bound on memory required for rendezvous is $\Omega (\log n)$ bits,
even for the line of length $n$. This lower bound matches the upper bound from \cite{CKP} valid for arbitrary graphs.

Our main positive result is a proof that
the amount of memory needed for rendezvous \emph{with simultaneous start} in trees depends essentially on the number $\ell$ of leaves of the tree, and is exponentially less impacted by the number
$n$ of nodes.
 Indeed, we show two identical agents with $O(\log \ell + \log\log n)$
bits of memory that solve the rendezvous problem in all trees with $n$ nodes and $\ell$ leaves.
Hence, for the class of trees with polylogarithmically many leaves, there is an exponential gap in minimum
memory size needed for rendezvous between the scenario
with arbitrary delay and the scenario with delay zero. 

Moreover, we show that the size $O(\log \ell + \log\log n)$ of memory needed for rendezvous is optimal, even in the class of trees with degrees bounded
by 3. More precisely, we prove two lower bounds. First,   for infinitely many integers $\ell$, we show a class of arbitrarily large trees with maximum degree 3
and with $\ell$ leaves, for which rendezvous with simultaneous start requires $\Omega (\log \ell)$ bits of memory. Second, we show
that $\Omega(\log \log n )$ bits of memory  are required for rendezvous with simultaneous start in the line of length $n$. 
These two bounds together imply that our upper bound $O(\log \ell + \log\log n)$ cannot be improved, even for the class of trees with maximum degree 3.

\subsection{Bibliographic note}

Note that our definition of solving the rendezvous problem is stronger than the definition 
used in the conference versions \cite{FP,FP2} of this paper. Indeed, rendezvous 
should occur \emph{for any port labeling}. As opposed to what is 
claimed in~\cite{FP2}, the exponential gap described in this paper 
does not carry over to the case where the ability of achieving rendezvous 
may depend on the port labeling.

More precisely, it was claimed in \cite{FP2} that the
positive result concerning the size $O(\log \ell + \log\log n)$ of memory for which rendezvous with simultaneous start is possible,
holds for arbitrary initial positions that are not symmetric with respect to a given port labeling $\mu$ of the tree in which agents operate.
This result is in fact incorrect in this formulation. Indeed, it has been recently proved in \cite{CKP2} that, for some 
port labeling of a line and some initial positions that are not symmetric {\em with respect to this labeling},
rendezvous with simultaneous start requires a logarithmic number of bits, while $\ell=2$ for the line. However, our positive result holds
for agents starting from arbitrary non perfectly symmetrizable initial positions. 
The algorithm and its analysis remain similar as in~\cite{FP2}. (The exact place where the provided arguments do not extend to the case where the ability of achieving rendezvous 
may depend on the port labeling will be pointed out to the reader).

On the other hand, all negative results from \cite{FP} and \cite{FP2}
hold in the present setting as well. 

\subsection{Related work}
\label{subsec:relatwork}

The rendezvous problem was first mentioned in \cite{schelling60}. Authors investigating  rendezvous (cf.
\cite{alpern02b} for an extensive survey) 
considered either the geometric scenario (rendezvous in an interval of the real line, see, e.g.,  \cite{baston98,baston01,gal99},
or in the plane, see, e.g., \cite{anderson98a,anderson98b}), or rendezvous in networks, see e.g., \cite{DFKP,TSZ07,YY}. Many papers, e.g., \cite{alpern95a,alpern02a,anderson90,baston98,israeli} study
the probabilistic setting: inputs and/or rendezvous strategies are random. 

A lot of effort has been dedicated to the study of the feasibility of rendezvous, and to the time required to achieve this task, when feasible. For instance, deterministic rendezvous with agents equipped with tokens used to mark nodes was considered, e.g., in~\cite{KKSS}. Deterministic rendezvous of agents equipped with unique labels was discussed in \cite{DGKKP,DFKP,KM}. (In this latter scenario, symmetry is broken by the use of the different labels of agents, and thus rendezvous is sometimes possible even for strongly symmetric
 initial positions of the agents). Recently, rendezvous using variants of Universal Traversal Sequences was investigated in~\cite{TSZ07}.
Surprisingly though, as opposed to what was done for the graph exploration problem (see, e.g., \cite{CR80,FI04,koucky01b,R08}), or for other tasks such as routing (see, e.g., \cite{FG02,FG01}), few papers were devoted to study the amount of memory required by the agents for achieving rendezvous. Up to our knowledge, the only existing results prior to the conference papers \cite{FP,FP2} on which the present paper is based were dedicated to rendezvous in rings. Memory needed for randomized rendezvous in the ring is discussed, e.g., in~\cite{KKPM08}. In the recent paper \cite{CKP} the authors showed that deterministic rendezvous can be solved in arbitrary
$n$-node graphs using $O(\log n)$ memory  bits (for arbitrary delay between starting times of the agents) and that this number of bits is necessary, even in rings and
even for simultaneous start. Tradeoffs between time of rendezvous in trees and the size of memory of the agents are studied in \cite{CKP2}. 
The impact of memory size on the feasibility of the related task of tree exploration, 
for trees with unlabeled nodes, has been studied in \cite{DFKP2,GPRZ}.

A natural extension of the rendezvous
problem is that of gathering \cite{fpsw,israeli,lim96,thomas92}, when more than two agents have to meet in one location.
In~\cite{YY} the authors considered rendezvous of many agents
with unique labels. 

Apart from the synchronous model used in this paper, several authors have investigated asynchronous rendezvous in the plane \cite{CFPS,fpsw} and in network environments
\cite{BIOKM,CLP,DGKKP}.
In the latter scenario the agent chooses the edge which it decides to traverse but the adversary controls the speed of the agent. Under this assumption rendezvous
in a node cannot be guaranteed even in very simple graphs, hence the rendezvous requirement is relaxed to permit the agents to meet inside an edge. 

\section{Framework and Preliminaries}\label{sec:preliminaries}

\subsection{Model}

We consider mobile agents traveling in trees with
locally labeled ports. The tree and its size are a priori unknown to
the agents. We first define precisely an individual agent. An agent is an abstract state machine 
$\caA=(S,\pi,\lambda,s_0)$, where $S$ is a set of states among which there is
a specified state $s_0$ called the {\em initial} state, $\pi:S\times
\mathbb{Z}^2 \to S$, and $\lambda:S\to \mathbb{Z}$. Initially the agent is at some node
$u_0$ in the initial state $s_0\in S$. The agent performs actions
in rounds measured by its internal clock. Each action can be either a move to an adjacent
node or a null move resulting in remaining in the currently occupied node. 
State $s_0$ determines a natural 
number $\lambda(s_0)$. If $\lambda(s_0)=-1$ then the agent makes a null move (i.e., remains at $u_0$).
If $\lambda(s_0)\geq 0$ then the agent leaves $u_0$ by port $\lambda(s_0)$ modulo the degree of $u_0$. When
incoming to a node $v$ in state $s\in S$, the behavior of the agent is as follows.
It reads the number $i$ of the port through which it entered $v$ and
the degree $d$ of $v$. The pair $(i,d)\in \mathbb{Z}^2$ is an input symbol
that causes the transition from state $s$ to state
$s'=\pi(s,(i,d))$. If the previous move of the agent was null,
(i.e., the agent stayed at node $v$ in state $s$) then the pair $(-1,d)\in \mathbb{Z}^2$ is the input symbol
read by the agent, that causes the transition from state $s$ to state
$s'=\pi(s,(-1,d))$. 
In both cases $s'$ determines an integer
$\lambda(s')$, which is either $-1$, in which case the agent makes a null move, or
a non negative integer indicating a port number 
by which the agent leaves $v$ (this port is $\lambda(s')\bmod d$). The agent continues
moving in this way, possibly infinitely.

Since we consider the rendezvous problem for identical agents, we assume that agents are
copies $A$ and $A'$ of the same abstract state machine $\caA$, starting at two distinct nodes $v_A$ and $v_{A'}$, 
called the {\em initial positions}.
We will refer to such identical machines as a {\em pair of agents}. It is assumed that the internal clocks of
a pair of agents tick at the same rate. The clock of each agent starts when the agent starts executing its actions. 
Agents start from their initial position
with  {\em delay} $\theta \geq 0$, controlled by an adversary. This means that the later agent starts executing its
actions $\theta$ rounds after the first agent. Agents do not know which of them is first and what is the value of $\theta$. 
We seek agents with small memory, measured by the number of states of the corresponding automaton, or equivalently by the number of bits on which these states are encoded. An automaton with $K$ states requires $\Theta (\log K)$ bits of memory.

We say that a pair of agents solves the rendezvous problem {\em  with arbitrary delay} (resp.  {\em with simultaneous start}) 
 in a class  of trees, if, for any tree in this class, for any port labeling of this tree,
and for any initial positions that are not perfectly symmetrizable, both agents are eventually in the same node of the tree in the same round,
regardless of the starting rounds of the agents (resp. provided that they start in the same round). 

\subsection{Preliminary results}

Consider any tree $T$ and the following sequence of trees constructed recursively: $T_0=T$, and $T_{i+1}$ is the tree obtained from
$T_i$ by removing all its leaves. $T'=T_j$  for the smallest $j$ for which $T_j$ has at most two nodes. If $T'$ has one node, then
this node is called the {\em central node} of $T$. If $T'$ has two nodes, then the edge joining them is called the {\em central edge} of $T$.
A tree $T$ with a port labeling $\mu$ is called {\em symmetric}, if there exists a non-trivial automorphism $f$ of the tree  (i.e., an automorphism $f$ such that $f(u) \neq u$, for some $u \in V$)
preserving this port labeling.  
If a tree with port numbers has a central node, then it cannot be symmetric. 

We define the ``basic walk" starting at node $v$ the walk resulting from an agent performing the following actions: leave node $v$ by port $0$, and, perpetually, whenever entering a degree-$d$ node by port $i\in\{0,\dots,d-1\}$, leave that node by port $(i+1) \bmod d$. Of course, a basic walk can be bounded to perform for $t$ steps (instead of perpetually), in which case we refer to a basic walk of length $t$.  Note that a basic walk of length $2(n-1)$ in an $n$-node tree returns to its starting node. 

The following statement is an easy consequence of the techniques and results from \cite{GPRZ}.

\begin{fact}\label{fact}
There exists an agent accomplishing the following task in an arbitrary tree: using $O(\log m)$ 
bits of memory, it finds the number $m$ of nodes in the tree, returns and stops at its initial position, and detects whether the tree has a central node, or has a central edge but is not symmetric, or has a central edge and is symmetric. Moreover, 
\begin{itemize}
\item
if the tree has a central node $x$, then the agent finds the minimum number of steps of a basic walk
from its initial position to the central node $x$;
\item
if the tree has a central edge $e=\{x,y\}$ but is not symmetric, then, for every initial position, 
the agent finds the minimum number of steps of a basic walk
from its initial position to the \emph{same} extremity $x$ of the central edge;
moreover, it knows which port at this extremity corresponds to the central edge; 
\item
if the tree is symmetric, then the agent finds the minimum number of steps of a basic walk
from its initial position to the farthest extremity\footnote{Why the farthest and not the closest is for technical reasons that should appear clear further in the text.} of the central edge; moreover, it knows which port at this extremity corresponds to the central edge. 
\end{itemize}
\end{fact}

In the sequel, the procedure accomplishing the above task starting at node $v$ will be called Procedure {\tt Explo}$(v)$.

\section{Rendezvous with arbitrary delay}

It was proved in \cite{CKP} that rendezvous with arbitrary delay can be accomplished in arbitrary $n$-node graphs using $O(\log n)$ bits of memory. On the other hand, observe
that rendezvous requires $\Omega(\log n)$ bits of memory in arbitrarily large trees with $2n+1$ nodes and maximum degree $n$.
The lower bound examples are trees $T_n$ consisting 
of two nodes $u$ and $v$ of degree $n$, both linked to a common node $w$, and to $n-1$ leaves. However, these trees
have linear degree and the reason for the logarithmic memory requirement is simply that agents with smaller memory are incapable
of having an output function $\lambda$ with range of linear size, and thus the adversary can place one agent in node $u$, the other in a leaf adjacent to $v$, and distribute
ports in such a way that none of the agents can ever get to node $w$, which makes rendezvous infeasible, in spite of non perfectly symmetrizable initial positions. 

This example leaves open the question if rendezvous with sub-logarithmic memory is possible, e.g., in all trees with constant maximum degree.
It turns out that if the delay is arbitrary, this is not the case: rendezvous requires logarithmic memory even for the class of lines.

\begin{theorem}\label{delay}
Rendezvous with arbitrary delay in the $n$-node line requires agents with $\Omega(\log n)$ bits of memory. 
\end{theorem} 

\begin{proof}
Let $k$ be the number of memory bits of the agent and $K=2^k$ be its number of states.
Place one agent at some node $u$ of the infinite line where each edge has the same port number at its two extremities. In any interval of length $K+1$ there exist two nodes
at which the agent is in the same state. Let $x_1$ be the first node of the trajectory of the agent in which this happens
and let $s$ be the state of the agent at $x_1$.  Let $x_2$ be the second node of the trajectory of the agent 
at which the agent is in state $s$. Let $\delta$ be the distance between $u$ and $x_1$ 
and let $d$ be the distance between $x_1$ and $x_2$.

We construct the following instance of the rendezvous problem (see Fig.~\ref{fig:lowerbound}). The line is of length $8(K+1)+1$.
Let $e$ be the central edge of this line. Assign number 0 to ports leading to edge $e$ from both its extremities,
and assign other port labels so that ports leading to any edge at both its extremities get the same number 0 or 1.
(This is equivalent to 2-edge-coloring of the line). Let $z$ be the endpoint of the line, for which $x_1$ is between $z$ and $x_2$.
Let $y_1$ and $y_2$ be symmetric images of $x_1$ and $x_2$, respectively, according to the axis of symmetry of the line.
Let $y_0$ be the node  distinct from $y_2$, at distance $d$ from $y_1$. Let $v$ be the node at distance $\delta$ from $y_0$,
such that the vectors $[x_1,u]$ and $[y_0,v]$ have opposite directions. The other agent is placed at node $v$.

\begin{figure}[h]
\begin{center}
\includegraphics[width=0.9\linewidth]{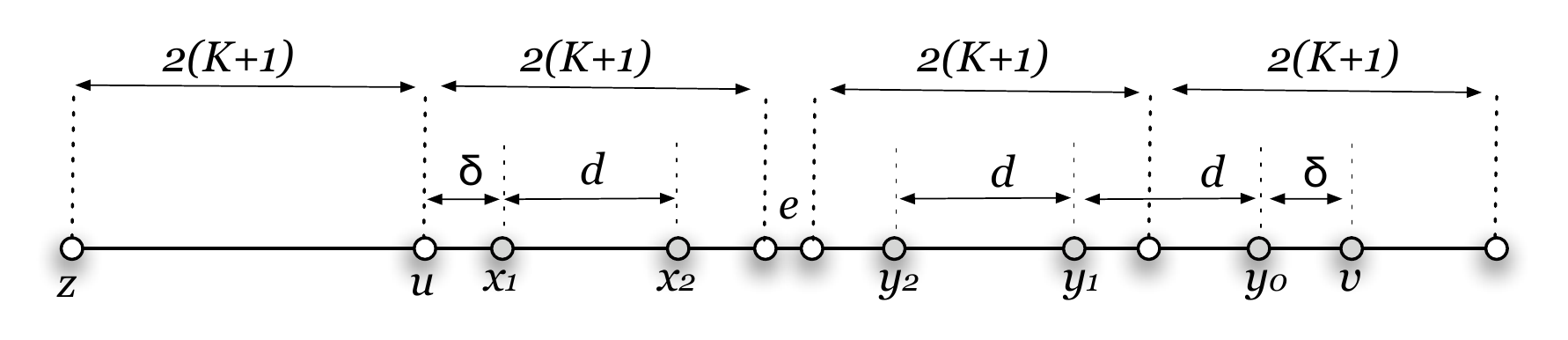}
\caption{Construction in the proof of Theorem~\ref{delay}.}
\label{fig:lowerbound}
\end{center}
\end{figure}

Let $t_1$ be the number of rounds  that the agent starting at $u$ takes to reach\footnote{We say that the agent {\em reaches} node $v$ in state $s$, if $s$ is the state in which the agent leaves $v$,
i.e., it leaves $v$ by port $\lambda(s).$} $x_1$ in state $s$.
Let $t_2$ be the number of rounds  that the agent starting at $v$ takes to reach  $y_1$ in state $s$. Let $\theta=t_2-t_1$.
The adversary delays  the agent starting at $u$ by $\theta$ rounds. Hence the agent starting at $u$ reaches $x_1$ at the same
time $t$ and in the same state as the agent starting at $v$ reaches $y_1$. The points $x_1$ and $y_1$ are symmetric positions, hence rendezvous is impossible after time $t$. Before time $t$ the two agents were on different sides of edge $e$,
in view of $\delta +d \leq 2(K+1)$, hence rendezvous did not occur, 
although the initial positions of the agents are not a perfectly symmetrizable pair. The size of the line is $O(K)=O(2^k)$, which concludes the proof.
\end{proof}

Together with the logarithmic upper bound from \cite{CKP},  the above result completely solves the problem of determining the minimum memory
of the agents permitting rendezvous with arbitrary delay. Hence in the rest of the paper we concentrate on rendezvous with simultaneous start,
thus assuming that the delay $\theta=0$.

\section{Rendezvous with simultaneous start}

\subsection{Upper bound}
\label{sec:alg}

It turns out that the size of memory needed for rendezvous with simultaneous start depends on two parameters of the tree:
the number $n$ of nodes and the number $\ell$ of leaves. In fact we show that rendezvous in trees with $n$ nodes and $\ell$ leaves
can be done using only $O(\log \ell + \log\log n)$ bits of memory. Thus, for trees with polylogarithmically many leaves, $O(\log\log n)$
bits of memory are enough. In view of Theorem \ref{delay}, this shows an exponential gap in the minimum memory size
needed for rendezvous between the scenarios with arbitrary delay and with delay zero.

\begin{theorem} \label{theo:upperbound}
There is a pair of identical agents solving rendezvous with simultaneous start in all trees, and using, for any integers $n$ and $\ell$, $O(\log \ell + \log\log n)$ bits of memory in trees with at most $n$ nodes and at most $\ell$ leaves.
\end{theorem}

The rest of the section is dedicated to the proof of Theorem~\ref{theo:upperbound}. Let $T$ be any tree, and let $\mathbf{v}$ and $\mathbf{v'}$ be the initial positions of the two agents in $T$. Let $T'$ be the \emph{contraction} of $T$, that is the tree obtained from $T$ by replacing every path\footnote{Here, by \emph{path} we mean a sequence of adjacent nodes of degree~2, all pairwise distinct.} in $T$ joining two nodes of degree different from~2 by an edge (the ports of this edge correspond to the ports at both extremities of the contracted path). Notice that if $T$ has $\ell$ leaves, then its contraction $T'$ has at most $2\ell-1$ nodes. 

Our rendezvous algorithm uses Procedure  {\tt Explo}, defined in Section~\ref{sec:preliminaries}, as a subroutine. More precisely, each of the two agents  executes procedure {\tt Explo}  in $T$, ignoring the degree-2 nodes. That is, protocol {\tt Explo} is modified so that whenever an agent enters a degree-2 node through port $i\in\{0,1\}$ in some state $s$, it will leave that node in the next round by port $(i+1) \bmod 2$, in the same state $s$. In fact, the are some subtle additional details in the modified version of {\tt Explo}, when the initial node is of degree different from~2. Specifically, let $s_0$ be the initial state of an agent executing {\tt Explo}.  Our modified agent starts in an additional state $s^*_0$. If the initial node $\mathbf{v}$ has a degree different from~2, then it enters state $s_0$ and starts {\tt Explo}$(\mathbf{v})$, ignoring the degree-2 nodes. Otherwise, the agent remains in state $s^*_0$ and leaves the initial node through port~0. The agent then performs a basic walk, remaining in state $s^*_0$, until it enters a node of degree $1$ (i.e., a leaf of the tree $T$). At such a node, denoted  by $\mathbf{v}_{leaf}$, the agent enters state $s_0$ and starts {\tt Explo}$(\mathbf{v}_{leaf})$, ignoring the degree-2 nodes. We call {\tt Explo-bis} the procedure   {\tt Explo}  modified in this way. Observe that, in trees with no nodes of degree~2, the two protocols  {\tt Explo}  and {\tt Explo-bis} are executed identically. Hence, protocols  {\tt Explo}  and  {\tt Explo-bis} are executed identically in $T'$. Formally, for an initial position $\mathbf{v}$, let  us define 
\[\mathbf{\widehat{v}} = \left \{ \begin{array}{ll}
\mathbf{v} & \mbox{if $\deg(v)\neq 2$}\\
\mathbf{v}_{leaf} & \mbox{otherwise}
\end{array}\right.\]
Then, the following holds.

\begin{claim}\label{claim:identique} 
Once an agent  starting from some node $\mathbf{v}$ has reached node $\mathbf{\widehat{v}}$, the states at nodes of degrees different from~$2$ of the agent performing {\tt Explo-bis} in $T$ are identical to the states of an agent  performing {\tt Explo} in $T'$ starting from node $\mathbf{\widehat{v}}$.
\end{claim}

Using this claim, rendezvous in $T$ is achieved as follows. 

\paragraph{Stage 1.} Each of the two agents  executes procedure {\tt Explo-bis} from their respective initial positions $\mathbf{v}$ and $\mathbf{v'}$. 

After having completed {\tt Explo-bis}, each agent knows whether the contraction tree $T'$ is symmetric or not. (It is non-symmetric if either there is a central node, or there is a central edge and the two port-labeled trees obtained by removing the central edge in $T'$ are not isomorphic --- the isomorphism must preserve both the structure of the trees, and the port labelings). 

\paragraph{Stage 2.} The nature of the second stage differs according to whether $T'$ is symmetric or not.

In the non symmetric case, the rendezvous protocol uses Fact~\ref{fact}, which states that the two agents performing Procedure {\tt Explo} will eventually identify a single node $x$ of $T'$. Node $x$ is identified by the number of steps of the basic walk performed in $T'$ to reach that node from the initial position. Notice that, although {\tt Explo} ensures (by Fact~\ref{fact}) that each agent returns to its initial position $\mathbf{v}$ after completing the procedure, Claim~\ref{claim:identique} guaranties only that the agent applying {\tt Explo-bis}  returns to a node $\mathbf{\widehat{v}}$. Nevertheless, this is sufficient, since the length of the basic walk reaching $x$ is the length of the one starting from node $\mathbf{\widehat{v}}$, ignoring degree-2 nodes. Note that this length does not exceed twice the number of edges of $T'$, and thus it can be encoded on $O(\log \ell)$ bits. 

Therefore, each of the agents act as follows: 

\begin{itemize}
\item If there is a central node $x$ in $T'$, then Rendezvous is achieved by waiting for the other agent at that node. 
\item Similarly, if there is a central edge in $T'$, and the tree $T'$ is not symmetric, then let $x$ be the extremity of the central edge of $T'$ identified by protocol {\tt Explo-bis}; rendezvous is achieved by waiting for the other agent at that node. 
\end{itemize}

The difficult and more challenging situation is when the contraction tree $T'$ has a central edge with two non distinguishable extremities, in which case the ability to solve the rendezvous problem depends on the large tree $T$ and on the initial positions of the two agents in $T$. Achieving rendezvous is complicated by the constraint that the agents must use sub-logarithmic memory when $\ell$ is small. The main part of the proof will be dedicated to describing how this task can actually be achieved in a memory efficient manner. 

\paragraph{Sub-stage 2.1.} (for the case when $T'$ symmetric) Resynchronization.

Recall that we are in a situation where each of the two agents has performed  {\tt Explo-bis}. An agent starting from node $\mathbf{v}\in T$ has not necessarily returned to node $\mathbf{v}$, but to node $\mathbf{\widehat{v}}\in T'$.  Each agent executes Procedure {\tt Synchro} defined as follows. It starts the execution of a basic walk in $T$, leaving the current node $\mathbf{\widehat{v}}$ by port~0. This basic walk will end when the agent is back at node $\mathbf{\widehat{v}}$. This is simply insured by counting the number of edge-traversals in $T'$: the agent stops the basic walk after $2(\nu-1)$ edge-traversals in $T'$, where $\nu$ denotes the number of nodes in $T'$. Since $\nu\leq 2\ell-1$, counting up to $O(\nu)$ does not require more that $O(\log\ell)$ bits. The basic walk proceeds with the following insertions: at each visited node $w$ with degree different from~2 (i.e., at each node of $T'$), the agent performs  {\tt Explo-bis}$(w)$, except for the very last node of $T'$ visited by the basic walk, that is except when the agent returns, for the last time, at its initial position $\mathbf{\widehat{v}}$.

Since agents performing Procedure {\tt Synchro} starting from different initial positions $\mathbf{\widehat{v}}$ execute identical actions, only in different order, we have the following: 

\begin{claim}\label{claim:delaybeta}
Two agents starting simultaneously at arbitrary initial positions $\mathbf{v}$ and $\mathbf{v'}$ in $T$ finish Procedure {\tt Synchro} with a delay $\beta=|L-L'|$ where $L$ (resp., $L'$) is the length of the basic walk in $T$ leading from $\mathbf{v}$ to $\mathbf{\widehat{v}}$ (resp.,  from $\mathbf{v'}$ to $\mathbf{\widehat{v}'}$).
\end{claim}

Once the agents are resynchronized (their desynchronization is now precisely  $\beta$), each of them proceeds to the second part of Stage~2. 

\paragraph{Sub-stage 2.2.} (for the case when $T'$ symmetric) Rendezvous in a virtual line.

After the execution of Procedure {\tt Synchro}, the agent with initial position $\mathbf{v}$ is back at $\mathbf{\widehat{v}}$. In view of Fact~\ref{fact}, since it has applied {\tt Explo}$(\mathbf{\widehat{v}})$ at the very beginning of the rendezvous protocol, the agent knows the number of steps of the basic walk from $\mathbf{\widehat{v}}$ to the farthest extremity of the central edge of $T'$. So, its first action in Sub-stage 2.2 is to go to this node, following a basic walk. We denote by $\mathbf{\widehat{v}}_{far}$ (resp., $\mathbf{\widehat{v}'}_{far}$) the farthest extremity of the central edge of $T'$ reached by the agent starting from $\mathbf{v}$ (resp., from $\mathbf{v'}$). 

Since the contraction tree $T'$ is symmetric, the two agents may end up in two different nodes of $T$, i.e., possibly $\mathbf{\widehat{v}}_{far}\neq \mathbf{\widehat{v}'}_{far}$. For instance, in the $n$-node path with an odd number of edges, the two agents may end up in the two extremities of the path. Also, in the binomial tree with $n$-nodes (cf.\cite{CLR}), the two agents may end up in the two roots of the two binomial subtrees of $T$ with $n/2$ nodes. Still, we prove that rendezvous is possible with little memory assuming that the two initial positions of the agents were not perfectly symmetrizable in $T$. Actually, the first of the two key ingredients in our proof is showing how rendezvous can be achieved in the path (or line) using agents with $O(\log\log n)$ bits of memory. 

In the lemma below, we consider \emph{blind} agents in paths, that is agents that ignore port labels. More precisely, when entering a node, such an agent can just distinguish between the incoming edge and the other edge (if any). Let $P=(v_1,\dots,v_m)$ be an $m$-node path, and consider two identical blind agents initially located at nodes $v_a$ and $v_b$, $a<b$. Rendezvous using blind agents is possible if and only if $m$ is odd, or $m$ is even and $a-1\neq m-b$. Of course, a standard agent can simulate the behavior of a blind agent. When applying the lemma below with standard agents, we will make sure that the starting positions  $v_a$ and $v_b$ are such that rendezvous is achievable even with blind agents. 

\begin{lemma} \label{lem:rdvchemin}
There exists a pair of identical blind agents accomplishing rendezvous with simultaneous start in all paths, whenever it is possible, and using $O(\log\log m)$ bits of memory in paths with at most $m$ nodes.
\end{lemma}

\begin{proof}
Let $P=(v_1,\dots,v_m)$ be an $m$-node path, and consider two identical blind agents initially located at nodes $v_a$ and $v_b$, $a<b$. To achieve rendezvous, the two agents perform a sequence of traversals of $P$, executed at lower and lower speeds, aiming at eventually meeting each other at some node. More precisely, for an integer $s\geq 1$, a traversal of the path is performed at \emph{speed} $1/s$,  if the agent remains idle $s-1$ rounds before traversing any edge. For instance, traversing $P$ from $v_1$ to $v_m$ at speed $1/s$ requires $(m-1)s$ rounds. Our rendezvous algorithm for the line, called \prem, performs as follows. 
  
\begin{center}
\fbox{
\begin{minipage}{10cm}
\textbf{Begin}\\
\hspace*{3ex} start in arbitrary direction; \\
\hspace*{3ex} move at speed 1 until reaching one extremity of the path; \\
\hspace*{3ex} $p\leftarrow 2$;\\
\hspace*{3ex} {\bf While} no rendezvous {\bf do} \\
\hspace*{6ex} traverse the entire path twice, at speed $1/p$; \\
\hspace*{6ex} $p\leftarrow$ smallest prime larger than $p$; \\
\textbf{End}
\end{minipage}
}
\end{center}

We now prove that, whenever rendezvous is possible for blind agents (i.e., when $m$ odd, or $m$ even and $a-1\neq m-b$), the two agents meet before the $p$th iteration of the loop, for $p=O(\log n)$. Let $p_j$ be the $j$th prime number ($p_1=2$). Hence the speed of each agent at the $j$th execution of the loop is $1/p_j$. If rendezvous has not occurred during the $j$th execution of the  loop, then the two agents have crossed the same edge, say  $e=\{v_c,v_{c+1}\}$, at the same time $t$, in opposite directions. This can occur if, for instance, the agent initially at $v_a$ moves to node $v_1$, traverses twice the path at successive speeds $p_1,\dots,p_{j-1}$, and,  $c \, p_j $ rounds after having eventually started walking at speed $p_j$, traverses the edge $e$ at time $t$,   while the other agent initially at $v_b$ moves to $v_m$, traverses twice the path at successive speeds $p_1,\dots,p_{j-1}$, and,  $(m-c) p_j $ rounds after having eventually started walking at speed $p_j$, traverses the same edge $e$ in the other direction at the same time $t$. In fact, there are four cases to consider, depending on the two starting directions of the two agents: towards $v_1$ or towards $v_m$. From these four cases, we get that one of the following four equalities must hold (the first one corresponds to the previously described scenario: $v_a$ moves towards $v_1$ while $v_b$ moves towards $v_m$):
\begin{itemize}
\item $t = (a-1) + 2(m-1) \sum_{i=1}^{j-1}p_i +  c \, p_j  = (m-b) + 2 (m-1) \sum_{i=1}^{j-1}p_i  + (m-c) p_j$
\item $t = (a-1) + 2(m-1) \sum_{i=1}^{j-1}p_i +  (m-1) p_j  + (m-c)p_j = (b-1) + 2 (m-1) \sum_{i=1}^{j-1}p_i  + c\, p_j$
\item $t = (m-a) + 2(m-1) \sum_{i=1}^{j-1}p_i +  (m-c) p_j  = (m-b) + 2 (m-1) \sum_{i=1}^{j-1}p_i  + (m-1) p_j + c\, p_j$
\item $t = (m-a) + 2(m-1) \sum_{i=1}^{j-1}p_i +  (m-c) p_j  = (b-1) + 2 (m-1) \sum_{i=1}^{j-1}p_i  + c\, p_j$
\end{itemize}
Therefore we get that 
$p_j \; \mbox{divides} \; |a-b|, \;\mbox{or} \;  p_j \; \mbox{divides} \; |m-(a+b)+1|$.
As a consequence, since the $p_i$'s are primes, we get that if the two agents have not met after the $j$th execution of the loop, then 
\[ \prod_{i \in {\cal I}} p_i \; \mbox{divides} \; |a-b| \;\;\;\mbox{and} \;\;\;  \prod_{i \in {\cal J}} p_i \; \mbox{divides} \; |m-(a+b)+1|\]
where ${\cal I} \cup  {\cal J} = \{1,\dots,j\}$. Therefore, since the $p_i$'s are primes, $\prod_{i=1}^{j} p_i  \; \mbox{divides} \; |a-b|\cdot |m-(a+b)+1|.$ Hence, if rendezvous is feasible, it must occur at or before the $j$th execution of the loop, where $j$ is the largest index such that  $\prod_{i=1}^{j} p_i  \; \mbox{divides} \; |a-b|\cdot |m-(a+b)+1|$. Thus it must occur at or before the $j$th execution of the loop, where $j$ is the largest index such that  $\prod_{i=1}^{j} p_i \leq m^2$. 

Let $\pi(x)$ be the number of prime numbers smaller than or equal to $x$. On the one hand, we have $\prod_{i=1}^{j} p_i \geq 2^{\pi(p_j)}$. Hence, rendezvous must occur at or before the $j$th execution of the loop, where $j$ is the largest index such that $2^{\pi(p_j)}\leq m^2$, i.e., $\pi(p_j)\leq 2 \log m$. On the other hand, from the Prime Number Theorem we get  that $\pi(x)\sim x/\ln(x)$, i.e., $\lim_{x\rightarrow \infty} \frac{\pi(x)}{x/\ln(x)}=1$. Hence, for $m$ large enough, $\pi(x) \geq  x/(2\ln(x))$. Thus rendezvous must occur at or before the $j$th execution of the loop, where $j$ is the largest index such that $p_j/\ln p_j \leq 4\log m$. 

From the above, we get that (1) rendezvous must occur whenever it is feasible, and (2) it occurs at or  before the $j$th execution of the loop, where $\log p_j \leq O(\log\log m)$. Since the next prime $p$ can be found using $O(\log p)$ bits, e.g., by exhaustive search, we get that \prem\/ performs rendezvous using agents with $O(\log\log m)$ bits of memory. 
\end{proof}

The (blind) agents described in Lemma~\ref{lem:rdvchemin} perform a protocol called \prem. This protocol uses the infinite sequence of prime numbers. We denote by $\prem(i)$ the protocol \prem\/ modified so that it stops after having considered the $i$th prime number. 

We now come back to our general rendezvous protocol in trees (with port numbers). Let $\nu=2x$ be the number of nodes in the contraction tree $T'$. (We have $\nu$ even, since $T'$ is symmetric with respect to its central edge). We define a (non-simple) path called the \emph{rendezvous path}, denoted by $P$, that will be used by the agents to rendezvous using protocol \prem. To define $P$, let $u$ and $v$ be the two extremities of the path in $T$ corresponding to the central edge in $T'$. We have $\{\mathbf{\widehat{v}}_{far},\mathbf{\widehat{v}'}_{far}\}\subseteq\{u,v\}$. The path $P$ is called the \emph{central path}, and is denoted by $C$. Abusing notation, $C$ will also be used as a shortcut for the instruction: ``traverse $C$". 

Let \BW\/ (for ``basic walk") be the instruction of performing the following actions: leave by port $0$, and, perpetually, whenever entering a degree-$d$ node by port $i\in\{0,\dots,d-1\}$, leave that node by port $(i+1) \bmod d$. Similarly, let \CBW\/ (for ``counter basic walk"), be the instruction of performing the following: leave by the port used to enter the current node at the previous step, and, perpetually, whenever entering a degree-$d$ node by port $i$, leave that node by port $(i-1)\bmod d$. For $j\geq 1$, let $\BW(j)$ (resp., $\CBW(j)$) be the instruction to execute \BW\/ (resp., \CBW) until $j$ nodes of degree different from~2 have been visited. Let $B_u$ (resp., $B_v$) be the path corresponding to the execution of $\BW\big (2(\nu-1) \big)$ from $u$ (resp., from $v$). Note that a node can be visited several times by the walk, and thus neither $B_u$ nor $B_v$ are simple. Note also that since $T'$ has $\nu$ nodes, it has $\nu-1$ edges, and thus both $B_u$ and $B_v$ are closed paths, i.e., their extremities are $u$ and $v$, respectively. Let $\overline{B}_u$ (resp., $\overline{B}_v$) be the path corresponding to the execution of $\CBW\big(2(\nu-1)\big)$ from $u$ (resp., from $v$). We define 
\[P=(B_u\wedge C_{u\to v} \wedge \overline{B}_v \wedge C_{v\to u})^{5\ell} \wedge (B_u \wedge C_{u\to v} \wedge \overline{B}_v)\]
where ``$\wedge$" denotes the concatenation of paths, $C_{u\to v}$ (resp., $C_{v\to u}$) denotes the path $C$ traversed from $u$ to $v$ (resp., from $v$ to $u$), and, for a closed path $Q$, $Q^\alpha$ denotes $Q$ concatenated with itself $\alpha$ times. 

The path $P$ is well defined. Indeed, the sequence $B_u\wedge C_{u\to v} \wedge \overline{B}_v \wedge C_{v\to u}$ leads back to node $u$. Also, the two extremities of the path are $u$ and $v$. Now, the agents have no clue whether they are standing at $u$ or at $v$. Nevertheless, we have the following. 

\begin{claim}
Starting from an extremity $u$ or $v$ of the central path $C$, an agent performing the sequence of instructions 
$$\Big(\BW\big(2(\nu-1)\big),C,\CBW\big(2(\nu-1)\big),C\Big)^{5\ell}, \BW\big(2(\nu-1)\big), C, \CBW\big(2(\nu-1)\big)$$
traverses the path $P$ from one of its extremities to the other. 
\end{claim}

Before establishing the claim, note that instructions $\BW\big(2(\nu-1)\big)$ and $\CBW\big(2(\nu-1)\big)$ are meaningful, since agents can have counters of size $O(\log \ell)$ bits, and they know $\nu$ in view of Fact~\ref{fact}. To establish the claim, it suffices to notice that the path $\overline{P}$ reverse to $P$ is given by
\[\overline{P}=(B_v \wedge C_{v\to u} \wedge \overline{B}_u) \wedge (C_{u\to v} \wedge B_v \wedge C_{v\to u} \wedge \overline{B}_u)^{5\ell}
= (B_v\wedge C_{v\to u} \wedge \overline{B}_u \wedge C_{u\to v})^{5\ell} \wedge (B_v \wedge C_{v\to u} \wedge \overline{B}_u).\]

The two agents will use protocol \prem\/ along the path $P$ to achieve rendezvous. However, to make sure that rendezvous succeeds, the two agents must not start \prem\/ simultaneously at the two extremities of $P$, in order to break symmetry. Unfortunately, this requirement is not trivial to satisfy. Indeed, one can guarantee  some upper bound on the delay between the times the two agents reach the two extremities of $C$ (and thus of $P$ as well)  that does not exceed $n$, but no guarantee can be given for the minimum delay, which could be zero. This is because the delay does not depend on the tree $T'$, but on the tree $T$. Hence two agents starting simultaneously in $T$ may actually finish Stage~2.1 of our protocol (i.e., the execution of {\tt Synchro}) at the same time, even if $T$ is not symmetric, and even if $T$  is symmetric but the starting positions were not perfectly symmetrizable. The second key ingredient in our proof is a technique guaranteeing  eventual desynchronization of the two agents. A high level description of this technique is summarized in Figure~\ref{fig:phase2}. We describe this technique in detail below. 

\begin{figure}[t]
\begin{center}
\fbox{
\begin{minipage}{11cm}
\textbf{Begin}\\
\hspace*{3ex} {\bf for consecutive values} $i\geq 1$ {\bf do}  \;\; \emph{/* outer loop */}\\
\hspace*{6ex} \emph{/* try rendezvous */}\\
\hspace*{6ex} {\bf for} $j=0,1,\dots,2(\nu-1)$ {\bf do} \;\; \emph{/* first inner loop */}\\
\hspace*{9ex} perform $\BW(j)$; \\
\hspace*{9ex} perform $\CBW(j)$;  \;\; \emph{/* back to the original position */}\\
\hspace*{9ex} perform  $\prem(i)$ on the rendezvous path $P$;\\
\hspace*{6ex} \emph{/* reset */}\\
\hspace*{6ex} go to the other extremity of the central path $C$;\\
\hspace*{6ex} {\bf for} $j=0,1,\dots,2(\nu-1)$ {\bf do}  \;\; \emph{/* second inner loop */}\\
\hspace*{9ex} perform $\BW(j)$;\\
\hspace*{9ex} perform $\CBW(j)$;  \;\; \emph{/* back to the original position */} \\
\hspace*{6ex} return to the original extremity of the central path $C$;\\
\textbf{End}
\end{minipage}
}
\end{center}
\vspace{-3ex}
\caption{Second phase of the rendezvous (performed when the contraction tree is symmetric).}
\label{fig:phase2}
\end{figure}

The outer loop of the protocol in Figure~\ref{fig:phase2} states how many consecutive prime numbers the protocol will test while performing \prem\/ along the path $P$. Performing \prem$(i)$ for successive values of $i$, instead of just \prem, is for avoiding a perpetual execution of \prem\/ in the case when the two agents started the execution of phase~2 at the same time from the two extremities of $P$. For every number $i\geq 1$ of primes to be used in \prem, the protocol performs two inner loops. The first one is an attempt to achieve rendezvous along $P$, while the second one is used to upper bound the delay between the two agents at the end of the outer loop, in order to guarantee that the next execution of the outer loop will start with a delay between the two agents that does not exceed~$n$. 

During the first inner loop, an agent executing the protocol performs a series of basic walks, of different lengths. For $j=0$, the agent performs nothing. In this case,  $\prem(i)$ is performed on $P$ directly. For $j>0$,  the agent performs a basic walk in $T$ to the $j$th node of degree different from~2 that it encounters along its walk. When $j=2(\nu-1)$, the basic walk is a complete one, traversing each edge of $T$ twice. Each $\BW(j)$ is followed by a $\CBW(j)$, so as to come back to the original position at the same extremity of the path $P$. Once this is done, the agent performs $\prem(i)$ on $P$. 

The second inner loop aims at resetting the two agents. For this purpose, each agent goes to the other extremity of $C$, performs the same sequence of actions as the other agent had performed during its execution of the first inner loop, and returns to its original extremity of $C$. This enables resetting the two agents in the following sense. 

\begin{claim}\label{claim:reschedule}
Let $t$ and $t'$ be the times of arrival of the two agents at $\mathbf{\widehat{v}}_{far}$ and $\mathbf{\widehat{v}'}_{far}$ after the execution of {\tt Synchro}, respectively. Then the difference between the times the two agents enter each execution of the outer loop of the protocol in Figure~\ref{fig:phase2} remains identical, equal to~$|t-t'|$. 
\end{claim} 

To establish the claim, just notice that, during every execution of the outer loop, the sets of actions performed by the two agents inside the loop are identical, differing only by their orders. 

Note that we can express $|t-t'|=|(L+\widehat{L})-(L'+\widehat{L}')|$ where $L$ and $L'$ are defined in Claim~\ref{claim:delaybeta}, and $\widehat{L}$ (resp., $\widehat{L}'$) denotes the length of the basic walk leading from $\mathbf{\widehat{v}}$ (resp., $\mathbf{\widehat{v}'}$) to $\mathbf{\widehat{v}}_{far}$ (resp., to~$\mathbf{\widehat{v}'}_{far}$). A consequence of Claim~\ref{claim:reschedule} is the following lemma.

\begin{lemma}\label{lem:delay}
Let $t$ and $t'$ be the times of arrival of the two agents at $\mathbf{\widehat{v}}_{far}$ and $\mathbf{\widehat{v}'}_{far}$ after the execution of {\tt Synchro}, respectively. For every $i$, the delay between the two agents at the beginning of each execution of \prem$(i)$ cannot exceed $|t-t'| + 16 n\ell$. 
\end{lemma} 

\begin{proof}
For $j\geq 1$, let $l_j$ and $l'_j$ be the lengths (i.e., numbers of edges) of the paths in $T$ between the $(j-1)$th and the $j$th node of degree different from~2 that is met by the two agents, respectively, during their basic walk from their positions at the two extremities of $C$. At the $j$th iteration of the inner loop, one agent has traversed $2 \sum_{a=1}^j \sum_{b=1}^a l_b$ edges during $\BW(a)$ and $\CBW(a)$ for all $a=1,\dots,j$. The other agent has traversed  $2 \sum_{a=1}^j \sum_{b=1}^a l'_b$ edges during the same $\BW(a)$ and $\CBW(a)$. Since the number of rounds of \prem$(i)$ is the same for both agents, we get that their delay is at most: 
\begin{eqnarray*}
|t-t'|+2 \sum_{a=1}^j \sum_{b=1}^a | l_b - l'_b | & \leq & |t-t'| +4 (\nu-1) \sum_{b=1}^{2(\nu-1)} | l_b - l'_b | \\
 & \leq &  |t-t'| + 4 (\nu-1) \sum_{b=1}^{2(\nu-1)} \max\{l_b,l'_b\} \\
 & \leq &  |t-t'| + 8 (\nu-1) n  \\
 & \leq &  |t-t'| + 8 \nu n  \\
 & \leq &  |t-t'| + 16 n \ell. 
\end{eqnarray*} 
This completes the proof of the lemma. 
\end{proof}

\begin{lemma}\label{lem:delaybis}
Assume that the two agents have not met when they arrive at $\mathbf{\widehat{v}}_{far}$ and $\mathbf{\widehat{v}'}_{far}$ after the execution of {\tt Synchro}. For every $i$, if at the beginning of each execution of \prem$(i)$ the delay between the two agents is zero, then their initial positions were perfectly symmetrizable in $T$. 
\end{lemma} 

\begin{proof}
Fix $i\geq 1$, and assume that, at the beginning of each of the $2\nu-1$ executions of \prem$(i)$ in the outer loop, the delay between the two agents is zero. This implies that, using the same notations as in the proof of Lemma~\ref{lem:delay}, for every $j=0,\dots,2(\nu-1)$ we have
\[ t + 2 \sum_{a=1}^j \sum_{b=1}^a l_b  = t' +2 \sum_{a=1}^j \sum_{b=1}^a l'_b \; .  \]
Therefore, $$t=t'$$ and $$l_j=l'_j\; \mbox{for every $j=0,\dots,2(\nu-1)$}.$$ 
These equalities imply that the tree $T$ is topologically symmetric: there is an automorphism $f$ which extends the port preserving automorphism $f'$ of $T'$ mapping the two symmetric subtrees $T'_1$ and $T'_2$ of $T'$ hanging at the two extremities of the central edge of $T'$ ($f'$ induces an isomorphism between $T'_1$ and $T'_2$ preserving port labels). Indeed, since  the two agents have not met when both of them arrive at $\mathbf{\widehat{v}}_{far}$ and $\mathbf{\widehat{v}'}_{far}$, the fact that $t=t'$ implies that $\mathbf{\widehat{v}}_{far}\neq \mathbf{\widehat{v}'}_{far}$. We have $\mathbf{\widehat{v}'}_{far} = f'(\mathbf{\widehat{v}}_{far})$. More generally, if $x_j$ (resp., $x'_j$) denotes the $j$th node of $T'$ reached by the basic walk starting at $\mathbf{\widehat{v}}_{far}$ (resp., $\mathbf{\widehat{v}'}_{far}$), we have $x'_j=f'(x_j)$. By definition, $l_j$ (resp., $l'_j$)  is the length of the path in $T$ between $x_{j-1}$ and $x_j$ (resp., between $x'_{j-1}$ and $x'_j$). Since $l_j=l'_j$, we get that the number of degree-2 nodes in $T$ between $x_{j-1}$ and $x_j$ is the same as the number of degree-2 nodes in $T$ between $x'_{j-1}$ and $x'_j$. Thus $f'$ can be extended to match nodes of these two paths, preserving adjacencies. Since this holds for every $j$, we get that $T$ is topologically symmetric\footnote{The automorphism $f$ does not necessarily preserve the port numbers in $T$ along the paths between nodes with degree different from~2. This is the reason why, as opposed to what is claimed in~\cite{FP2}, the initial positions of the agents are not necessarily symmetric  in $T$. We show however that they are perfectly symmetrizable in $T$.}.

To sum up, the tree $T$ is topologically symmetric (by automorphism $f$), and its contraction tree $T'$ is symmetric (by automorphism $f'$, which preserves port labels). A consequence of this fact is the following crucial observation. Let us consider the following port labeling $\mu$. The port numbers at nodes of degree larger than~2 are the same as in $T'$. The port labeling is completed arbitrarily at nodes of degree 2, preserving the following condition: if $\{z,z'\}$ is an edge in $T$  with at least one extremity $z$ of degree~2, then the port number at $z$ corresponding to  $\{z,z'\}$  is equal to the port number at $f(z)$ corresponding to  $\{f(z),f(z')\}$. Two basic walks starting from two symmetric positions in $T'$ generate two sequences of nodes such that the $i$th nodes of the two sequences are symmetric in $T$ with respect to $\mu$. Indeed, the ``branching'' nodes, i.e., the nodes of degree at least~3, are symmetric, and basic walks are oblivious of the port numbers at nodes of degree at most~2. The same observation holds for counter basic walks. It also holds if the port number of the outgoing edge from the starting nodes are not~0, under the simple assumption that they are equal. 

We use the above observation to show that the two nodes $\mathbf{v}$ and $\mathbf{v'}$ are perfectly symmetrizable. Since $T'$ is symmetric, it is sufficient to show that $\mathbf{v}$ and $\mathbf{v'}$ are topologically symmetric. The two agents have reached nodes $\mathbf{\widehat{v}}_{far}$ and $\mathbf{\widehat{v}'}_{far}$ after procedure {\tt Synchro}, entering these nodes from the central path. Indeed, on the one hand, $\mathbf{\widehat{v}}_{far}$ and $\mathbf{\widehat{v}'}_{far}$ are the \emph{farthest} extremity of the central edge of $T'$ coming from  $\mathbf{\widehat{v}}$ and $\mathbf{\widehat{v}'}$, respectively, and, on the other hand, the basic walks reaching these nodes are of minimum length (cf., Fact~\ref{fact}). Since  $\mathbf{\widehat{v}}_{far}$ and $\mathbf{\widehat{v}'}_{far}$ are symmetric in $T'$, the port numbers of the edges incident to these nodes on the central path are identical. Let $i$ be this port number. Consider two counter basic walks of length $t=t'$ starting from $\mathbf{\widehat{v}}_{far}$ and $\mathbf{\widehat{v}'}_{far}$, leaving the starting node by port number $i$. These counter basic walks proceed backwards, first along the basic walk from $\mathbf{\widehat{v}}$ to $\mathbf{\widehat{v}}_{far}$ for $\widehat{L}$ steps, and next along the basic walk from  $\mathbf{v}$ to $\mathbf{\widehat{v}}$ for $L$ steps. If $\mathbf{\widehat{v}}=\mathbf{v}$ then $L=0$. If $\mathbf{\widehat{v}}\neq\mathbf{v}$, then the articulation between the two basic walks $\mathbf{v} \to \mathbf{\widehat{v}}$ and  $\mathbf{\widehat{v}} \to \mathbf{\widehat{v}}_{far}$   occurs at $\mathbf{\widehat{v}}=\mathbf{v}_{leaf}$. Since we have chosen this latter node as a \emph{leaf}, the sequence of basic walks $\mathbf{v} \to \mathbf{\widehat{v}}$ and  $\mathbf{\widehat{v}} \to \mathbf{\widehat{v}}_{far}$ is actually equal to a basic walk $\mathbf{v} \to \mathbf{\widehat{v}}_{far}$ of length $t=L+\widehat{L}$. Hence the counter basic walk of length $t$ starting from $\mathbf{\widehat{v}}_{far}$ by port $i$ leads  to the initial position $\mathbf{v}$. The same holds for the other walk of length $t'=t$. Therefore, $\mathbf{v}$ and $\mathbf{v'}$ are topologically symmetric, and thus they are perfectly symmetrizable.
\end{proof}

In view of the previous lemma, since $\mathbf{v}$ and $\mathbf{v'}$ are not perfectly symmetrizable, at each execution $i$ of the outer loop, there is an execution $j$ of $\prem(i)$ for which the two agents do not start the second phase at the same time from their respective extremities of $P$. Moreover, by Lemma~\ref{lem:delay}, during this $j$th execution of $\prem(i)$, the delay $\delta$ between the two agents is at most  $|t-t'| + 16 n\ell$. We have $|t-t'|=|(L+\widehat{L})-(L'+\widehat{L}')|$, where the four parameters are lengths of basic walks. These four basic walks have lengths at most $2(n-1)$. Hence, $|t-t'|\leq 4n$. Therefore, $\delta \leq 20 n\ell$. The length of the rendezvous path $P$ is larger than $20n\ell$ because $B_u$ and $B_v$ are each of length at least $2n$. Therefore, at the first time when both agents are simultaneously in the $j$th execution of $\prem(i)$, they occupy two non perfectly symmetrizable positions in $P$: one is at one extremity of $P$, and the other is at some node of $P$ at distance $\delta>0$ along $P$ from the other extremity of $P$. Moreover, since the delay $\delta$ between the two agents is smaller than the length of the path $P$, the agent first executing $\prem(i)$ has not yet  completed the first traversal of $P$ when the other agent starts $\prem(i)$. As a consequence, the two agents act as if $\prem(i)$  were executed with both agents starting simultaneously at non perfectly symmetrizable positions in the path. Now, for small values of $i$, $\prem(i)$ may not achieve rendezvous in $P$. However, in view of Lemma~\ref{lem:rdvchemin}, for some $i = O(\log n)$, rendezvous will be completed whenever the initial positions of the agents were not perfectly symmetrizable in $T$. 

We complete the proof by checking that each agent uses $O(\log \ell + \log\log n)$ bits of memory. Protocol {\tt Explo-bis} executed in $T$ consumes the same amount of memory as Protocol  {\tt Explo} executed in $T'$. Since $T'$ has at most $2\ell-1$ nodes,  {\tt Explo-bis} uses $O(\log \ell)$ bits of memory. During the second stage of the rendezvous, a counter is used for identifying the index $j$ of the inner loop. Since $j\leq 2\nu \leq 4\ell$, this counter uses $O(\log \ell)$ bits of memory. All executions of \prem\/ are independent, and performed one after the other. Thus, in view of Lemma~\ref{lem:rdvchemin}, a total of $O(\log\log n)$ bits suffice to implement these executions. The index $i$ of the outer loop grows until it is large enough so that $\prem(i)$ achieves rendezvous in a path of length $O(n\ell)$. Thus, $i\leq \log(n\ell)$, and thus $O(\log\log(n\ell))= O(\log\log n)$ bits suffice to encode this index. This completes the proof of Theorem~\ref{theo:upperbound}. 

\subsection{The lower bound $\Omega( \log\log n)$}

In this section we prove the lower bound $\Omega( \log\log n)$  on the size of memory required for rendezvous with simultaneous start in a $n$-node line.

\begin{theorem}\label{theo:lo}
Rendezvous with simultaneous start  in the $n$-node line requires agents with
$\Omega( \log\log n)$ bits of memory.
\end{theorem}

The rest of the section is dedicated to the proof of Theorem~\ref{theo:lo}. For proving  the theorem, note that we can restrict 
ourselves to lines whose edges are properly colored 1 and 2, so that the port numbers at the two extremities of an edge colored 
$i$ are set to $i$. In this setting, the transition function of an agent in a line is $\pi: S\times 
\{1,2\} \to S$ that describes the transition that occurs when an agent enters a node of degree $d\in\{1,2\}$ in state $s\in S$. In 
this situation, the agent changes its state to state $s'=\pi(s,d)$, and performs the action $\lambda(s')$. The fact that one does 
not need to specify the incoming port number is a consequence of the edge-coloring, which implies that whenever an agent leaves a 
node by port $i$, it enters the next node by port $i$ too. 

Let us fix two identical agents $A$ and $A'$, with finite state set $S$, and transition function $\pi$. Let $\pi': S\to S$ 
be the transition function applied at nodes of degree~2 of the edge-colored line, i.e., $\pi'(s)=\pi(s,2)$ for any $s\in S$. To 
$\pi'$ is associated its transition digraph, whose nodes are the states in $S$, and there is an arc from $s$ to $s'$ if and only 
if $s'=\pi'(s)$. This digraph is composed of a certain number of connected components, say $r$, each of them of a similar 
shape, that is a circuit with inward trees rooted at the nodes of the circuit. Let $C_1,\dots,C_r$ be the $r$ circuits 
corresponding to the $r$ connected components of the transition digraph, and let $\gamma$ be the least common multiple of the 
number of arcs of these circuits, i.e., $\gamma=\lcm(|C_1|,\dots,|C_r|)$. We prove that there is a line of length proportional 
to $2\gamma+|S|$ in which $A$ and $A'$ do not rendezvous. 

First, observe that if $A$ and $A'$ cannot go at arbitrarily large distance from their starting positions, say they go at maximum 
distance $D$, then they cannot rendezvous in a line of length $4D+4$. Indeed, if the initial positions are two nodes at  
distance $2D+1$, and at distance at least $D+1$ from the extremities of the line, then the ranges of activity of the two agents 
are disjoint, and thus they cannot meet (one edge is added at one extremity of the line to break the symmetry of the initial 
configuration). 

Thus from now on, we assume that both agents can go at arbitrarily large distance from their starting positions. 

For the purpose of establishing our result, place the two agents $A$ and $A'$ on two adjacent nodes $v_A$ and $v_{A'}$ of an 
infinite line (whose edges are properly colored). Let $e=\{v_A,v_{A'}\}$ be the edge linking these two nodes. 

\begin{itemize}
\item Let $t_0$ be large enough so that $A$ is at distance at least  $2\gamma+|S|$ from its starting position after $t_0$ steps. 
\end{itemize}

Since $t_0>|S|$, agent $A$ at time $t_0$ is in some state $s_i \in C_i$ for some $i\in\{1,\dots,r\}$.  In fact, since $|C_i|$ 
divides $\gamma$, agent $A$ has fully executed $C_i$ at least twice. 

We define the notion of \emph{extreme position} for a circuit $C$. Let $s,\pi'(s),\dots,\pi'^{(k)}(s)$ be a circuit, with 
$s=\pi'^{(k)}(s)$. Assume that agent $A$ starts in state $s$ from node $u_0$ at distance at least $k+1$ 
from both extremities of the 
line. After having performed $C$ exactly once, i.e., after $k$ steps, agent $A$ is at some node $u_k$, back in state $s$. Let 
$u_0,u_1,u_2,\dots,u_k$ be the $k+1$ non necessarily distinct nodes visited by $A$ while executing $C$. The  \emph{extreme 
position} for $C$ starting in state $s$ is the node $u_j$ 
satisfying $$\dist(u_0,u_j)= \dist(u_0,u_k)+\dist(u_k,u_j), \;\;\; \mbox{and} 
\;\;\; \dist(u_0,u_j)=\max_{0\leq \ell \leq k}\dist(u_0,u_\ell).$$ Let $u_i$ be the extreme position for $C_i$ starting 
in $s_i$, 
and let us define the following parameters: 

\begin{itemize}
\item $\tau$ is the first time step among the $|C_i|$ steps after step $t_0$ at which $A$ reaches $u_i$. 
\vspace{-1.5ex}
\item $x$ is the distance of agent $A$ at time $\tau$ from its original position, i.e., $x=\dist(u_i,v_A)$; 
\vspace{-1.5ex}
\item $\tau'=\tau+2\gamma$;
\vspace{-1.5ex}
\item $x'$ is the distance of agent $A'$ at time $\tau'$ from its original position $v_{A'}$.
\end{itemize}

Note that, by symmetry of the port labeling, and from the fact that $A$ and $A'$ are identical
and operate in an infinite line, the two agents are on the two 
different sides of edge $e$ at time $\tau$. Note also that, between times $\tau$ and $\tau'$, agent $A'$ keeps on going further 
away from its original position, by repeating the sequence of actions determined by the circuit $C_i$. Hence $x'\neq x$. Actually, 
we have $x'>x$. We can therefore consider  the following construction. 

\paragraph{Initial configuration of the agents.}

Let $\caL$ be the properly 2-edge-colored line of length $x+x'+1$, formed by $x$ edges, followed by one edge called $e$, and 
followed by $x'$ edges. The two agents $A$ and $A'$ are placed at the two extremities $v_A$ and $v_{A'}$ of $e$, the same way they 
were placed at the two extremities of $e$ in the infinite line used to define $x$ and $x'$. 
\medbreak

Since $x\neq x'$, the initial positions of agents are not perfectly symmetrizable. 
Nevertheless, we prove that the two agents never meet in $\caL$, and thus rendezvous
is not accomplished. The 
adversary imposes no delay between the starting times of the agents, i.e., 
they both start acting simultaneously from their respective 
initial positions. 

One ingredient used for proving that the two agents do not rendezvous is the following general result, that we state as a lemma for 
further reference. 

\begin{lemma}
{\bf (Parity Lemma)}Consider two (not necessarily identical) agents initially at odd distance in a tree $T$, that 
start acting simultaneously in $T$. Let $t\geq 1$. Assume that one agent stays idle $q$ times in the time interval $[1,t]$, while 
the other one stays idle $q'$ times in the same time interval. If $|q-q'|$ is even, then the two agents are at odd distance at 
step $t$. 
\end{lemma}

\begin{proof}
At any step, if one agent moves while the other one stays idle, then the parity of their distance changes. On the other 
hand, if both agents move or both stay idle, then the parity of their distance remains unchanged. Let $a$ be the number of 
steps in $[1,t]$ when both agents were idle simultaneously. Then the parity of the inter-agent distance changes exactly 
$(q-a)+(q'-a)$ times in the time interval  $[1,t]$. Since $|q-q'|$ is even, $q+q'$ is also even, and thus $(q-a)+(q'-a)$ is even 
too.  Thus the parity of the inter-agent distance is the same at time 1 and at time $t$.
\end{proof}

The Parity Lemma enables us to establish the following. 

\begin{lemma}\label{lem:beforetau}
The two agents $A$ and $A'$ do not meet during the first $\tau$ steps.
\end{lemma}

\begin{proof}
Since the agents perform the same sequence of actions in the time interval $[1,\tau]$, we get that, for any $t\leq \tau$, the two 
agents have remained idle the same mumber of times in the time interval $[1,t]$, and thus, by the Parity Lemma (with $q=q'$), they 
are at odd distance at step $t$, since they originally started at distance~1. In other words, the two agents remain permanently at 
odd distance during the time interval $[1,\tau]$. Thus they cannot meet during this time interval.  
\end{proof} 

At step $\tau$, the behavior of the two agents becomes different. Indeed, agent $A$ is reaching one extremity of $\caL$, while $A'$ 
is visiting a degree-2 node. 

We analyze the states of the two agents when they reach extremities of $\caL$ during the execution of their protocol. Assume that 
agent $A$ reaches  the extremities of $\caL$ at least $k\geq 1$ times. Let $\sigma_j$ be the state of agent $A$ when it reaches 
any of the two extremities of $\caL$ for the $j$th time, $1\leq j\leq k$. 

\begin{lemma}\label{lem:samestate}
Agent $A'$ reaches the extremities of $\caL$ at least $k$ times. Moreover, if $\sigma'_j$ is the state of agent $A'$ when it 
reaches any of the two extremities of $\caL$ for the $j$th time, $1\leq j\leq k$, then $\sigma'_j=\sigma_j$. 
\end{lemma}

\begin{proof}
First, let us consider the case $k=1$. After time $\tau$ (i.e., after the time when $A$ reaches one extremity of $\caL$, in state 
$\sigma_1$), agent $A'$ keeps on repeating the execution of circuit $C_i$. This leads $A'$ to eventually reach the other 
extremity of $\caL$. Recall that we have considered the behavior of $A$ after time $t_0$ when $A$ was in state $s_i \in C_i$, and 
that $\tau$ was defined as the first time step among the $|C_i|$ steps after step $t_0$ at which $A$ reaches the extreme position 
$u_i$ of $C_i$ starting at $s_i$. Since $\tau'=\tau+2\gamma$, and since $|C_i|$ divides $\gamma$, we get that agent $A'$ is in 
state $\sigma_1$ at time $\tau'$. Moreover,  since $|C_i|$ divides $\gamma$, $A'$ reaches the extreme position $u_i$ of $C_i$ at 
time $\tau'$, and therefore time $\tau'$ is the first time when $A'$ is at distance $x'$ from $e$. Therefore $\sigma'_1=\sigma_1$, 
and the lemma holds for $k=1$. 

For $k>1$, the proof is by induction on the number of times $j$ agent $A$ reaches an extremity of $\caL$, $j=1,\dots,k$. By the 
previous arguments, the result holds for $j=1$. When agent $A$ reaches an extremity of $\caL$ for the $j$th time, it is in state 
$\sigma_j$. By the induction hypothesis, when agent $A'$ reaches an extremity of $\caL$ for the $j$th time, it is also in state 
$\sigma'_j=\sigma_j$. Therefore, the configuration for $A$ and $A'$ between two consecutive hits of an extremity of $\caL$ is 
actually symmetric. As a consequence, $\sigma'_{j+1}=\sigma_{j+1}$, and the lemma holds. 
\end{proof}

After time $\tau$ the walks of the agents can be decomposed in two different types of subwalks. A \emph{traversal period} for an 
agent is the subwalk between two consecutive hits of two different extremities of $\caL$ by this agent. A \emph{bouncing 
period} for an agent is a subwalk (possibly empty) performed between two consecutive traversal periods. 
Roughly, a bouncing period for 
an agent is a walk during which the agent starts from one extremity of $\caL$ and repeats bouncing (i.e., leaving and going back) 
that extremity until it eventually starts the next traversal period. 

Globally, an agent starts from its original position, performs some initial steps ($\tau$ for $A$, and $\tau'$ for $A'$), and then 
alternates between  bouncing periods and  traversal periods. These periods are not synchronous between the two agents because 
there is a delay of $2\gamma$ between them. Nevertheless, by Lemma~\ref{lem:samestate}, if one agent bounces at one extremity of 
$\caL$ during its $k$th bouncing period, then the other agent bounces at the other extremity of $\caL$ during its $k$th 
bouncing period. Similarly, if one agent traverses $\caL$ during its $k$th traversal period, then the other agent traverses
$\caL$ in the opposite direction during its $k$th traversal period. In fact, Lemma~\ref{lem:samestate} 
guarantees that the two agents 
perform symmetric actions with a delay of $2\gamma$, alternating bouncing at the two different extremities of $\caL$, and 
traversing $\caL$ in two opposite directions. 

The following lemma holds, by establishing that whenever one agent is in a bouncing period, the two agents are far apart. 

\begin{lemma}\label{lem:bouncing}
The two agents $A$ and $A'$ do not meet whenever one of them is in a bouncing period. 
\end{lemma}

\begin{proof}
There is a delay of $2\gamma$ between the two agents. During such a period of time, an agent can travel a distance at most 
$2\gamma$. Also, during its bouncing period, an agent cannot go at distance more than $|S|$ from the extremity of the line where 
it is bouncing. On the other hand, by the definitions of $t_0$ and $\tau>t_0$, we have $x>2\gamma+|S|$, and thus 
$x'>2\gamma+|S|$ as well. Therefore, when one of the agents is in a bouncing period, the distance between the two agents is at 
least $2\gamma+|S|$, and thus they cannot meet.
\end{proof}

The following lemma holds, by using the fact that $\gamma$ is the least common  multiple of the circuit lengths in the transition digraph of the agents, and by applying the Parity Lemma. 

\begin{lemma}\label{lem:traversal}
The two agents $A$ and $A'$ do not meet when both of them are in a traversal period. 
\end{lemma}

\begin{proof}
When both agents are in a traversal period, they started their period in the same state, from Lemma~\ref{lem:samestate}. Hence, 
they are eventually both performing the same circuit of states $C_i$. This occurs after the same initial time of duration at most 
$|S|$. This time corresponds to the time it takes to reach the circuit $C_i$ from the initial state at which the agents started 
their traversal period. As we already observed in the proof of Lemma~\ref{lem:bouncing}, since $x'>x>2\gamma+|S|$, the two agents 
are far apart during the transition period before both of them have entered the circuit $C_i$ executed during the considered 
traversal. Thus we can now assume that the two agents are performing $C_i$, traversing the line in two opposite directions. We 
prove that they cross along an edge, and hence they do not meet. Since the delay between the two agents is $2\gamma$ and since 
$\gamma$ is a multiple of $|C_i|$ for any $i\in\{1,\dots,r\}$, the delay is an even multiple of the length of the circuit $|C_i|$ 
performed at this traversal. As a consequence, at any step of their traversal periods, 
the number of times one agent was idle when the 
other was not, is even. The Parity Lemma with $|q-q'|=2\gamma/|C_i|$ then insures that the distance between the two agents remains 
odd during the whole traversal period. Thus they do not meet.
\end{proof}

\medbreak
\noindent{\bf Proof of Theorem~\ref{theo:lo}.} 
The two agents start an initial period that lasts $\tau$ steps. By Lemma~\ref{lem:beforetau} they do not meet during this 
period. Then the two agents alternate  between bouncing periods and traversal periods. By Lemma~\ref{lem:bouncing}, they do not 
meet when one of the two agents is in a bouncing period. When the two agents are in a traversal period, Lemma~\ref{lem:traversal} 
guarantees that they do not meet. Hence the two agents never meet, in spite of starting from non perfectly symmetrizable positions,
and thus they do not rendezvous in $\caL$. 
By the construction of the line $\caL$ and the setting of $\gamma$, we get that  $\caL$ is of length $O(|S|^{|S|})$. Therefore, rendezvous with simultaneous start in lines of size at most $n$ requires agents with at least $\Omega(\log\log n)$ memory bits. 
\hfill $\Box$ 

\subsection{The lower bound $\Omega(\log \ell)$}

In this section we prove that rendezvous with simultaneous start in trees with $\ell$ leaves requires 
$\Omega(\log \ell)$ bits of memory, even in the class of trees with maximum degree 3. Together with the lower bound of
$\Omega( \log\log n)$ on memory size needed for rendezvous in the $n$-node line\footnote{Notice that  the lower bound $\Omega( \log\log n)$ 
 holds for $n$-node trees of maximum degree 3 with many leaves as well: it suffices to attach identical binary trees 
on each extremity of the line, and the argument from the previous section goes through.}
established in Theorem~\ref{theo:lo}, this result proves
that our upper bound $O(\log \ell + \log\log n)$ from Section \ref{sec:alg} cannot be improved even for trees of maximum degree 3.

\begin{theorem}
For infinitely many integers $\ell$, there exists an infinite family of trees with $\ell$ leaves, for which rendezvous with simultaneous start requires
$\Omega(\log \ell)$ bits of memory.
\end{theorem}

\begin{proof}
Consider an integer $\ell=2i$, for any even $i$. Consider an $(i+1)$-node path with a distinguished endpoint called the root. 
To every internal node $x$ of the path attach either a new leaf, or a new node $y$ of degree 2 with a new leaf $z$ attached to it.
There are $2^{i-1}=2^{\ell/2-1}$ possible resulting non-isomorphic rooted trees. Call them {\em side trees}. Note that non-isomorphic is meant here
without the port-preserving clause: there are so many rooted trees which cannot be mapped to each other by {\em any} isomorphism, not only by
any isomorphism preserving port numbering. Fix an arbitrary port labeling in every side tree. 

For any pair of side trees $T'$ and $T''$ and for any positive even integer $m$, consider the tree $T$ consisting of side trees $T'$ and $T''$
whose roots are joined by a path of length $m+1$ (i.e., there are $m$ added nodes of degree two). Ports at the added nodes of degree two
are labeled as follows: both ports at the central edge have label~0, and ports at both ends of any other edge of the line have the 
same label~0 or~1. (This corresponds to a 2-edge-coloring of the line). 
Call any tree resulting from this construction a {\em two-sided tree}. Any such tree has $\ell$ leaves and maximum degree 3.
For any two-sided tree consider initial positions of the agents at nodes $u$ and $v$ of the joining path adjacent to roots of its side trees.

Consider agents with $k$ bits of memory (thus with $K=2^k$ states).
A {\em tour} of a side tree associated with an initial position ($u$ or $v$) is the part of the trajectory of the agent in this side tree between
consecutive visits of the associated initial position. Observe that the maximum duration $D$ of a tour is smaller than  
$K\cdot(3 i)$. Indeed, the number
of nodes in a side tree is at most $3i-1$, hence the number of possible pairs (state, node of the side tree) is at most
$K\cdot (3i-1)$. A tour of longer duration than this value would cause the agent to leave the same node twice in the same state,
implying an infinite loop. Such a tour could not come back to the initial position.

For a fixed agent with the set $S$ of states and a fixed side tree, we define the function $p: S \to S$ as follows.
 Let $s$ be the state in which the agent starts a tour. Then $p(s)$ is the state in which the agent finishes the tour. 
 Now we define the function $q: S \to S \times \{1,\dots,D\}$,
called the {\em behavior function},
by the formula $q(s)=(p(s),t)$, where $t$ is the number of rounds to complete the tour when starting in state $s$. The number of possible  behavior functions  is at most $F=(KD)^K$. A behavior function depends on the side tree for which it is constructed.

Suppose that $k \leq \frac{1}{3}\log \ell$. We have $D<3Ki=\frac{3}{2}K\ell$, hence $KD<\frac{3}{2}K^2\ell$. Hence we have
$\log K + \log\log (KD) \leq k+ \log\log(\frac{3}{2}K^2\ell) \leq k+2+\log k +\log\log \ell,$
which is smaller than $\frac{2}{3}\log \ell$ for sufficiently large $k$. It follows that $K\log(KD)<\ell ^{2/3}<\ell /2-1$, which implies
$F=(KD)^K<2^{\ell /2-1}$. Thus the number of possible behavior functions is strictly smaller than the total number of side trees.
It follows that there are two side trees $T_1$ and $T_2$ for which the corresponding behavior functions are equal.

Consider two instances of the rendezvous problem for any length $m+1$ of the joining line, where $m$ is a positive even integer: one in which both side trees are equal to $T_1$,
and the other for which one side tree is $T_1$ and the other is $T_2$. Rendezvous is impossible in the first instance because in this
instance initial positions of the agents form a symmetric pair of nodes with respect to the given port labeling. 
Consider the second instance, in which the initial positions of the
agents do not form a perfectly symmetrizable pair. Because of the symmetry of labeling of the joining line, agents cannot meet inside any of the side
trees. Indeed, when one of them is in one tree, the other one is in the other tree. Since the behavior function associated with side
trees $T_1$ and $T_2$ is the same, the agents leave these trees always at the same time and in the same state. Hence they cannot meet on the line, in view of its odd length and symmetric port labeling. This implies that they never meet, in spite of initial positions that are not perfectly symmetrizable. Hence rendezvous in 
the second instance requires $\Omega(\log \ell)$ bits of memory.
\end{proof}

\bibliographystyle{plain}


\end{document}